\documentclass[aps, pra, superscriptaddress, reprint, floatfix, showkeys]{revtex4-2}

\RequirePackage[l2tabu,orthodox]{nag}

\usepackage[utf8]{inputenc}

\usepackage{amsfonts, amsmath, amsthm, amssymb} 
\usepackage{mathtools}

\usepackage[caption=false]{subfig}

\usepackage{tikz}

\usetikzlibrary{automata, positioning}

\usepackage{standalone}

\usepackage{graphicx} 

\definecolor{myurlcolor}{rgb}{0,0,0.7}
\definecolor{myrefcolor}{rgb}{0.8,0,0}

\usepackage[
	breaklinks,
	pdftex,
	pdfauthor={Maciej Stankiewicz, Karol Horodecki, Roberto Salazar},
	pdftitle={The rank of contextuality},
	pdfsubject={Quantum Information Theory},
	pdfkeywords={contextuality, quantum},
	pdfproducer={Latex with hyperref},
	pdfcreator={pdflatex},
	colorlinks=true, 
	linkcolor=myrefcolor,
	citecolor=myurlcolor,
	urlcolor=myurlcolor
]{hyperref}

\usepackage{orcidlink}

\newtheorem{theorem}{Theorem}[section]
\newtheorem{lemma}{Lemma}[section]

\newtheorem{definition}{Definition}[section]

\newtheorem{axiom}{Axiom}
\newtheorem{remark}{Remark}[section]

\definecolor{ppblue}{RGB}{46,117,182}
\definecolor{ppred}{RGB}{197, 90, 17}

\def\>{\rangle}
\def\<{\langle}

\begin{document}

\title{The rank of contextuality}

\author{Karol Horodecki \orcidlink{0000-0001-7540-4147}}
\email[Correspondence: ]{karol.horodecki@ug.edu.pl}
\affiliation{Institute of Informatics, National Quantum Information Centre, Faculty of Mathematics, Physics and Informatics, University of Gdańsk, Wita Stwosza 57, 80-308 Gdańsk, Poland}
\affiliation{International Centre for Theory of Quantum Technologies, University of Gdańsk, Wita Stwosza 63, 80-308 Gdańsk, Poland}

\author{Jingfang Zhou}
\affiliation{Department of Applied Physics, Graduate School of Engineering, University of Tokyo, 7-3-1 Hongo, Bunkyo-ku, Tokyo 113-8656, Japan}

\author{Maciej Stankiewicz \orcidlink{0000-0001-8288-3860}}
\affiliation{International Centre for Theory of Quantum Technologies, University of Gdańsk, Wita Stwosza 63, 80-308 Gdańsk, Poland}

\author{Roberto Salazar \orcidlink{0000-0003-1737-1433}}
\affiliation{Faculty of Physics, Astronomy and Applied Computer Science, Jagiellonian University, 30-348 Kraków, Poland}

\author{Pawe\l{} Horodecki \orcidlink{0000-0002-3233-1336}}
\affiliation{International Centre for Theory of Quantum Technologies, University of Gdańsk, Wita Stwosza 63, 80-308 Gdańsk, Poland}
\affiliation{Faculty of Applied Physics and Mathematics, National Quantum Information Centre, Gdańsk University of Technology, Gabriela Narutowicza 11/12, 80-233 Gdańsk, Poland}

\author{Robert Raussendorf \orcidlink{0000-0003-4983-9213}}
\affiliation{Department of Physics and Astronomy, University of British Columbia, Vancouver, BC V6T 1Z1, Canada}

\author{Ryszard Horodecki \orcidlink{0000-0003-2935-290X}}
\affiliation{International Centre for Theory of Quantum Technologies, University of Gdańsk, Wita Stwosza 63, 80-308 Gdańsk, Poland}

\author{Ravishankar Ramanathan \orcidlink{0000-0003-1119-8721}}
\affiliation{Department of Computer Science, The University of Hong Kong, Pokfulam Road, Hong Kong}

\author{Emily Tyhurst \orcidlink{0000-0002-7903-4056}}
\affiliation{Department of Physics, University of Toronto, Toronto, ON, Canada}
\affiliation{Department of Physics and Astronomy, University of British Columbia, Vancouver, BC V6T 1Z1, Canada}

\date{May 20, 2022}
	
\begin{abstract} 
    Quantum contextuality is one of the most recognized resources in quantum communication and computing scenarios.
    We provide a new quantifier of this resource, the rank of contextuality (RC). 
    We define RC as the minimum number of non-contextual behaviors that are needed to simulate a contextual behavior. 
    We show that the logarithm of RC is a natural contextuality measure satisfying several properties considered in the spirit of the resource-theoretic approach. 
    The properties include faithfulness, monotonicity, and additivity under tensor product.
    We also give examples of how to construct contextual behaviors with an arbitrary value of RC exhibiting a natural connection between this quantifier and the arboricity of an underlying hypergraph.
    We also discuss exemplary areas of research in which the new measure appears as a natural quantifier.
\end{abstract}


\maketitle

		
\section{Introduction}
\label{sec:introduction}

Quantum contextuality is one of the oldest genuinely quantum phenomena \cite{Kochen-Specker, Joshi_bound, Peres_Incompatible, Mermin_quantum_probability} which has been intensively investigated in recent years (see \cite{CReview} for a recent review).
A number of its applications, which we will discuss shortly in what follows, make it a resource in the context of classical and quantum information processing \cite{Chitambar2019}.
This phenomenon is characterized by a set of partially commuting observables with property $\mathcal{S}$. 
Namely, they cannot be globally described by a distribution of value assignments of the outcomes of respective measurements before these measurements happen to be performed.
The values observed depend on the context in which given observables are measured, i.e., a subset $\mathcal{C}\subset\mathcal{S}$ of mutually commuting observables. 
In the specific case where commutation relations are imposed by spatial separation of the observables, contextuality is known as non-locality. 
This type of contextuality has attracted loads of attention, reaching practical ideas in the realm of device-independent information processing \cite{Bell-nonlocality}.

The research on contextuality took the direction of finding low dimensional examples of this phenomenon (exhibited by a relatively small number of observables, each with a low number of outcomes) \cite{Cabello_exper_testable, Yu2012, Klyachko2008}.
Further, a natural question was: is contextuality a resource, like non-locality? 
As we have noted, contextuality is a resource in several applications.
Just to mention a few there is among them the zero-error channel capacity \cite{CLMW}, quantum device-independent cryptography \cite{Context_crypto}, magic-state based quantum computation \cite{Howard_etal, BermejoVega2017}, quantum computation using shallow circuits \cite{Bravyi2018} and one-way communication \cite{Saha2019}, and capacity of multiple access channels \cite{Leditzky2020}.
For further development based on these seminal results, we refer the reader to \cite{Amaral2, Amaral8} and the review \cite{CReview}.

Another question is: how can one quantify contextuality? 
Four major ways to achieve this goal are: the contextuality fraction \cite{Abramsky_contextualfrac}, relative entropy of contextuality \cite{Grudka_contextuality}, robustness of contextuality \cite{Amaral_Geo}, and memory of contextuality \cite {Kleinmann_memory}. 
We focus on the latter quantity in this manuscript and define it in a novel way. 
In \cite{Kleinmann_memory} the memory is a resource connected to contextuality needed to simulate the behavior of a system. 
A specific scenario of the Peres-Mermin Game \cite{Peres_Incompatible, Mermin_quantum_probability} was addressed there. 
The memory cost is counted as the logarithm of the number of states of a finite automaton, which simulates a given set of observables. 
The simulation reproduces observable outputs based on sequences of inputs to an (in general probabilistic) automaton. 
The generated outputs satisfy the constraints of the Peres-Mermin Game \cite{Kleinmann_memory} or even reproduce the measurement outcomes of the interrogated observables acting on the internal quantum state \cite{Fagundes}.

\subsection{The scenario and main results} 
\label{sec:scenario}

In the formalism of \cite{Karol_Axiomatic}, the problem becomes as follows. 
We have a given set of partially commuting observables. 
Given a quantum state, by measuring subsets of mutually commuting observables on that state, we generate a behavior, which we call \textit{challenging behavior} $C$. 
We imagine a \textit{simulator} $\mathcal{S}$  as a machine (finite automaton) that is built from a certain number of non-contextual behaviors, i.e., obtainable in nature without the use of the set of contextual observables (we consider them as free). 
We consider these behaviors as the \textit{states of memory} of the machine.

We say that the machine simulates our challenging behavior if, for any context, when asked to measure it, the automaton outputs a sequence of symbols that are drawn from a distribution of that context, according to the definition of $C$.
In such a case, the machine is operationally indistinguishable from interrogating the behavior $C$ directly. 

Let us note here that the simulator given in \cite{Fagundes} does not satisfy this last property: it is easy to verify from the frequency of output (given a sufficiently long sequence) that the output string is not compatible with independent measurements of Peres-Mermin observables on a quantum state.

In the Theory of Finite Automata, the number of states is a quantifier of the needed memory to capture the behavior of a system. 
In our case, the states correspond to distinct non-contextual behaviors. 
Following the formalism of \cite{Karol_Axiomatic}, we propose the logarithm of the number of the automaton states to measure the \textit{memory cost} for simulating the contextuality of a given set of observables.

Let us stress here that this approach does not measure \textit{total} memory cost of simulating the set of observables.
In particular: the memory cost of a non-contextual set of quantum observables is by definition zero (see Figure \ref{fig:aut} for an illustration of this behavior).
One could, of course, fine-grain the outlook and measure the memory cost of specific non-contextual behaviors. 
Or use completely different ``states'' of memory, e.g., \textit{classical distribution of context}. 
In what follows, we will propose two natural axioms which may serve as a reason for our choice of measure. 
The first one states that any contextuality measure should be zero for non-contextual behaviors. 
The second extends it by saying that the measure on contextual behaviors has to be larger than any value of this measure on the non-contextual one. 
The latter axiom rules out simulations by contexts, as explained above.

Our second contribution, in the spirit of the resource theoretic approach of \cite{Acn2015, Karol_Axiomatic}, is the demonstration that the (log\nobreakdash-)rank of contextuality satisfies several of the axioms proposed in \cite{Karol_Axiomatic}.
It includes faithfulness, monotonicity under the broad class of operations (including some type of processing via so-called wirings \cite{Cabello_wirings}), and additivity on the tensor product of behaviors. 
We also support the fact that the (log\nobreakdash-)rank of contextuality is a novel measure by comparing it with other previously known measures.
We then provide an upper bound on the rank of contextuality based on a graph-theoretic quantity called \textit{arboricity} \cite{graph}. 
The latter quantity is the number of distinct forests (possibly disconnected acyclic graphs) that one needs to split a graph so that every edge belongs to one forest. 
A forest in our approach corresponds to a non-contextual behavior and a collection of such behaviors allows us to simulate the challenging behavior. 

Finally, we discuss possible applications of the measures in two research areas. 
First, we propose to describe \textit{probabilistic databases} (PDBs) \cite{Sa2019} employing our contextuality simulator $\mathcal{S}$, and we argue how the rank of contextuality would allow quantifying the storage overhead due to noise in strongly correlated attributes of the database \cite{ZhuW2004}. 
Furthermore, our analysis of the minimum number of non-contextual behaviors necessary to simulate a PDB would allow us to describe the minimum repairs of unclean PBDs for the above class of noises.
The second is the attacks by a non-signaling adversary on the device amplifying the private randomness of the so-called Santha-Vasirani source \cite{SV, Wojewodka2017, Colbeck2012}.

The remainder of this manuscript is organized as follows.
In Section \ref{sec:preliminaries}, we give introductory information.
In Section \ref{sec:logrank}, we define the (log-)rank of contextuality and discuss its connections to the automaton that simulates behaviors.
In Section \ref{sec:lr_is_measure}, we show that the log-rank is a contextuality measure and that it fulfills some valuable properties. 
In Section \ref{sec:aditivity}, we prove that log-rank admits additivity under the tensor product operations and briefly discuss log-rank for a mixture of behaviors.
In Section \ref{sec:unboundedness}, we provide examples of calculating the log-rank for some classes of graphs and behaviors. 
Furthermore, we show an achievable upper bound of the rank of contextuality based on graph arboricity.
In Section \ref{sec:comparison}, we compare log-rank with other, previously known, measures of contextuality. 
In Section \ref{sec:applications}, we discuss two possible applications of our measure of contextuality.
Finally, in Section \ref{sec:discussion}, we summarize our results and provide some open questions for further research.

\section{Preliminaries}
\label{sec:preliminaries}

With any set of observables $V=\{V_i\}$, we can identify their commutation relations and represent them by a hypergraph $H$ on this set where hyperedges are subsets of the powerset of $V$ that denote contexts, i.e., sets of mutually commuting observables. 
A measurement selection $X\in E$ is an \textit{input}. 
We refer to \textit{output} $A$ given input $X$ as an ordered list of outcomes $(a_{1}, a_{2}, \ldots, a_{|X|}))$, that is, it assigns a value to every observable in the context corresponding to $X$. 

The set of all conditional probability distributions denoted as $P(A \mid X)$, such that $\sum_{A} P(A \mid X)=1$ for every $X$, defines a convex polytope $\mathcal{P}$.
The polytope is additionally constrained by the following consistency condition: 
\begin{equation}
    \begin{split}
        \mathop\forall_{X, Y \in E } 
        \mathop\forall_{\{ C_k \}_{k \in X \cap Y}} 
        \sum_{\{ A_i \}_{i \in X \setminus Y}} P(\{ A_i \}, \{ C_k \} \mid X)\\ 
        = \sum_{\{ B_j \}_{j \in Y \setminus X}} P(\{ C_k \}, \{ B_j \} \mid Y) \label{consistencypairwise}
    \end{split}
\end{equation}
where $A_i, B_j, C_k$ denotes outputs of observable $i, j, k$ respectively.
Hence it is called the \textit{consistency polytope}.

We refer to a given point in polytope $P(A \mid X)$ as a \textit{behavior}, with input $X$ and output $A$. 
Each behavior in $\mathcal{P}$ can be either contextual or non-contextual. 
By non-contextual we mean behaviors that have a so-called non-contextual hidden variable model, so that 
\begin{equation}
    P_{NC}(A \mid X) = \sum_i p_i D_i(A \mid X),
    \label{NCbehavior}
\end{equation}
where $D_i(A \mid X)$ are deterministic behaviors, that is, for every context $X$ there exists
unique $A=a$, which is the output with probability $1$. 
Moreover, there exists a deterministic joint distribution for all the observables, $J = \sum_j q_j D_j (A \mid V)$ where $D_j(A \mid V)$ is a behavior that outputs with probability $1$ some fixed vector of measurement outcomes $(A_{1}, \ldots, A_{|V|})$ of all the observables from $V$.
In other words: a non-contextual behavior $P_{NC}$ is a probabilistic mixture of deterministic behaviors that have fixed values of outcomes to any observable beforehand. 
The set of all non-contextual behaviors forms a polytope within $\mathcal{P}$, called non-contextual polytope $\mathcal{N}$. 
Any behavior in $\mathcal{P} \setminus \mathcal{N}$ is called contextual. 

The combination of two behaviors, $P_{1}, P_{2}$ is given by a tensor product, $P_{1} \otimes P_{2}$, which is defined by the set of all products of the probability distributions of $P_{1}, P_{2}$. 
I.e., if $P_{1} = \{ q_{\alpha}\}, P_{2}= \{p_{\beta}\}$ then $P_{1} \otimes P_{2} = \{p_{\alpha} q_{\beta}\}$.

\subsection{Two models of simulation}

In this section, we describe, in detail, two models of simulating sets of observables.
The first one has been studied in the literature in the context of memory of contextuality \cite{Kleinmann_memory}. 
We will refer to it as \textit{observable-interrogation model}. 
The second, called here \textit{context-interrogation model} is intimately connected to a tomographic scenario.

In the observable-interrogation model, the verifier checks if a given system is contextual by asking a device to measure observables given in sequence.
We expect that in steps $i=1, 2, \ldots$, the description of an observable $A_i$ is given as an input to a device.
The device has a simulator inside (a prover). 
The simulator is required to reproduce data in the same manner as if a certain physical system was inside instead.
The system to be reproduced is represented by two entities.
First is a quantum state $\rho$.
The second is a (possibly infinite) physical realization of the observables from the sequence $A_1, A_2, \ldots$. 
The outputs of the device should correspond to subsequent measurements: $A_1$ on $\rho$, $A_2$ on the output state of the measurement
of $A_1$ on $\rho$, i.e., $A_2(A_1(\rho))$ and so on. 
When $A_i$ is asked twice in a row, the device should respond with the same output twice. 
An automaton choosing at random from several deterministic automatons that meet the requirements mentioned above has been given in \cite{Fagundes}.

A tomographic model that we also term here the context-interrogation model is much simpler as it requires less from the prover. 
At each step $i=1,2,\ldots$ (possibly infinite), a description of a context $X_i$ is set to a device, that has a simulator inside. 
The simulator, made by the prover, should result in outputs in a manner as if it contains a system that works in the following way:
In each step, a new copy of a fixed quantum state $\sigma$ is selected, all observables from the context $X$ are measured, and generate the outputs of the simulation device.
We remark that experimentalists often use the above simple model in the lab to verify the contextuality of a set of observables.

In \cite{Kleinmann_memory} the number of states of the previously described automaton (which can be in general probabilistic \cite{Fagundes}) determines the contextuality of the simulated device.

\section{The (log-)rank of contextuality as memory of a simulating automaton}
\label{sec:logrank}

In the tomographic scenario described in the previous section, as a measure of memory cost of simulation for a given behavior $C$, we propose the following: \textit{The number (and logarithm of that number) of distinct NC behaviors that are needed to simulate $C$}. 
Formally it reads:

\begin{definition}[Rank of contextuality]
    Let $P$ be a behavior.
    Let $\mathcal{N}$ be the set of all non-contextual behaviors, each of the form $N_i(A \mid X)$.
    Then we define the rank of contextuality denoted as $RC$ in the following way
    \begin{equation}
        \label{def:mem-cost}
        \begin{split}
            & RC(P)  \coloneqq  \min \bigg\{|S| : S \subset \mathcal{N}, \\
            & \left. \mathop\forall_{X=i} \quad \mathop\exists_{N(A \mid X) \in S} \quad \mathop\forall_{A} \quad N(A \mid X=i) = P(A \mid X=i) \right\},
        \end{split}
    \end{equation}
    where $|S|$ denotes cardinality of the set $S$.
\end{definition}

We will also define the logarithmic version of the rank of contextuality in the following way. 
    
\begin{definition}[Log-rank]
    Let $P$ be a behavior that has the rank of contextuality $RC(P)$.
    Then the log-rank denote as $RC_2$ is given by
    \begin{equation}
        RC_2(P) \coloneqq \log_2 RC(P).
    \end{equation}
\end{definition}

Let us discuss why we can view this number as a quantifier of memory.
As an initial example, we address the case of Peres-Mermin square game (see Figure \ref{fig:PM}).
We exemplify an automaton that simulates the PM game (see Figure \ref{fig:aut}).
The extension of this example for the general case of contextual behavior is straightforward.
More formally, a deterministic automaton is a tuple 
\begin{equation} 
    \mathbb{T}\coloneqq\<S,\Sigma,q_0,\delta,F\>
    \label{automaton}
\end{equation}
where $S$ is the set of \textit{states} and $\Sigma$ is an \textit{alphabet of the inputs to the automaton}. 
The state $q_0 \in S$ is the initial state of the automation and $\delta:S\times \Sigma\rightarrow S$ is a \textit{transition function}, which given state and an input symbol outputs a state. 
The set $F \subset S$ is the set of \textit{final states}.

In our case $S$ is the set of non-contextual behaviors and a symbol $\omega \in \Sigma$ which denotes finishing of the use of the automaton: $S \coloneqq \{N_1,\ldots ,N_n,q_0,F\}$. 
The set $\Sigma$ is the set of contexts $C$ of a contextual behavior that is simulated by $\mathbb{T}$. 
The transition function $\delta$ is defined as follows: if a pair $(N_i,C)$ is such that $C \in N_i$, the next state is $\delta(N_i,C) \coloneqq N_i$. 
If it is not the case, then the next state is any fixed $\delta(N_i,C) \coloneqq N_j$ such that
$C\in N_j$. 
The initial state $q_0$ and final state $F$ are added artificially, so that $\delta(q_0,C)=N_k$ such that $C \in N_k$ for any $k$ and $\delta(N_l,f)=F$ for any state $N_l$. 

Following \cite{Fagundes}, the memory of the automaton $\mathbb{T}$ described above can be naturally taken as $\log n$, i.e., the logarithm of the number of its states, where we do not take into account the artificially introduced initial and final states. 
The simplicity of this model lies in the fact that it does not allow for repetitive measurements. 
A context measured twice can yield different outcomes. 
This model, however, is powerful enough to measure the contextual inequalities that in the case of the Peres-Mermin scenario are equivalent to average values of the observables from all the contexts \cite{Peres_Incompatible, Mermin_quantum_probability}:
\begin{equation}
    \<C_1\>+\<C_2\>+\<C_3\>+\<R_1\>+\<R_2\>+\<R_3\>\leq 5.
\end{equation}

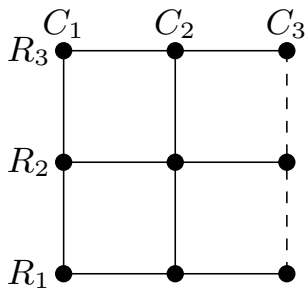
\begin{figure}[htbp]
    \begin{tikzpicture}
  \draw (0,0) -- (0,2);
  \draw (1,0) -- (1,2);
  \draw[dashed] (2,0) -- (2,2);
  \draw (0,0) -- (2,0);
  \draw (0,1) -- (2,1);
  \draw (0,2) -- (2,2);
  
  \filldraw[black] (0,0) circle (2pt) node[anchor=east]{$R_1$};
  \filldraw[black] (0,1) circle (2pt) node[anchor=east]{$R_2$};
  \filldraw[black] (0,2) circle (2pt) node[anchor=east]{$R_3$} node[anchor=south]{$C_1$};
  \filldraw[black] (1,0) circle (2pt);
  \filldraw[black] (1,1) circle (2pt);
  \filldraw[black] (1,2) circle (2pt) node[anchor=south]{$C_2$};
  \filldraw[black] (2,0) circle (2pt);
  \filldraw[black] (2,1) circle (2pt);
  \filldraw[black] (2,2) circle (2pt) node[anchor=south]{$C_3$};
  
\end{tikzpicture}
    \caption{
        \label{fig:PM}
        The hypergraph of the Peres-Mermin game. 
        The $6$ contexts are divided into two sets -- of columns $C_1,C_2,C_3$ and rows $R_1,R_2,R_3$. 
        The solid line represents the even distribution of outputs and the dashed line the odd one.
    }
\end{figure}

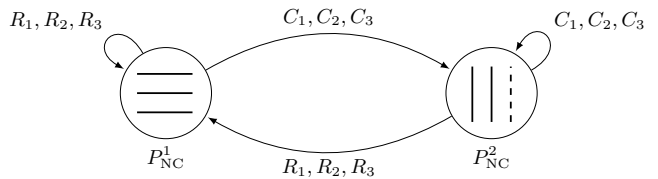
\begin{figure}[htbp]
    \begin{tikzpicture}[>=latex, shorten >=1pt,
node distance=2in, on grid, auto]
\node[state, label=below:$P_{\mathrm{NC}}^1$] (q0) {
    \begin{tikzpicture}[scale=0.3]
      \draw[thick] (0,0) -- (3,0);
      \draw[thick] (0,1) -- (3,1);
      \draw[thick] (0,2) -- (3,2);
    \end{tikzpicture}
};
\node[state, label=below:$P_{\mathrm{NC}}^2$] (q1) [right=of q0] {
    \begin{tikzpicture}[scale=0.3]
      \draw[thick] (0,0) -- (0,3);
      \draw[thick] (1,0) -- (1,3);
      \draw[thick, dashed] (2,0) -- (2,3);
    \end{tikzpicture}
};
\path[->] (q0) edge [bend left] node {$C_1, C_2, C_3$} (q1);
\path[->] (q1) edge [bend left] node {$R_1, R_2, R_3$} (q0);
\draw[->] (q0) to [out=120,in=150,looseness=5] node[above left] {$R_1, R_2, R_3$} (q0);
\draw[->] (q1) to [out=30,in=60,looseness=5] node[above right] {$C_1, C_2, C_3$} (q1);

\end{tikzpicture}

    \caption{
        \label{fig:aut}
        The finite automaton which is simulating the PM uniform behavior. 
        Contexts represented by solid lines are even distributions, and dashed are odd distributions. 
        Given input $X$ from alphabet $\{C1, C2, C3, R1, R2, R3\}$ it outputs symbols from a distribution taken from a context $X$ of a non-contextual behavior described by its current ``state'' of memory: either $P_{NC}^1$ or $P_{NC}^2$. 
        The arrows denote the change of states (loops -- staying in the same state), while the symbols above them rule under which letter of the alphabet the change of the state of memory happens. 
        The $P_{NC}^1$ is defined such that its rows are as those of the Peres-Mermin behavior, while the columns are the products of 3 maximally mixed 1-bit distributions. 
        The behavior $P_{NC}^2$ is defined analogously: the columns are defined as such as the corresponding columns of the Peres-Mermin behavior, while rows have maximally mixed distribution.
    }
\end{figure}

\section{Log-rank is a Contextuality Measure}
\label{sec:lr_is_measure}

We now discuss the criterion necessary for log-rank to be a contextuality measure and, more importantly, set up the framework for a resource theory surrounding it. 
In that respect, we set as an axiom
\begin{axiom}[Axiom 0]
    \label{axiom:0}
    Any measure of memory cost of behavior simulation due to its contextuality should be $0$ for non-contextual behaviors. 
\end{axiom}

It is clearly in case of our measure $RC_2$: the non-contextual behaviors the non-contextual behavior simulates itself, so the minimal $|S|$ equals $1$ in Eq.\ (\ref{def:mem-cost}). 
In the case that a contextuality measure does not satisfy Axiom \ref{axiom:0}, the following axiom is sufficient:
\begin{axiom}[Static Monotonicity]
    \label{axiom:monotonicity}
    Any measure $M$ of memory cost of simulation of a behavior due to its contextuality
	\begin{equation}
	    \mathop\forall_{\substack{P \in \mathcal{N}\\C \in \mathcal{P}\setminus \mathcal{N}}} M(P) \leq M(C).
	\end{equation}
\end{axiom}

This assumption rules out the simulation, which as the ``states'' of an automata's memory use just classical distributions of the contexts.
If it was not for the above axiom, we could design a deterministic behavior, which has more (exactly $5$) distinct distributions of contexts than the Peres-Mermin one (that has $2$ of them see Figure \ref{fig:PM}), e.g., a deterministic (hence non-contextual) behavior
\begin{equation}
    \begin{bmatrix}
        0 & 0 & 0 \\
        0 & 0 & 1 \\
        1 & 1 & 1 \\
    \end{bmatrix}
\end{equation} 
where $x$ takes the values of rows or columns. 
There would be also a non-contextual behavior which has only $1$ type of distribution $D(a \mid x)$ defined as $ \forall_x \delta_{a, 0}$, i.e., behavior outputting deterministically all zeros.
It implies that a measure that does not satisfy Axiom \ref{axiom:monotonicity} would not reflect contextuality as a reason for the cost of memory of the simulation.

Let us note that the (log-)rank of contextuality also satisfies Axiom \ref{axiom:monotonicity}. 
Further, one might wonder if the (log-)rank of contextuality is monotonic under standard resource theoretic operations \cite{Grudka_contextuality, Karol_Axiomatic}. 
Since the \textit{free resources} are explicitly non-contextual, (log-)rank has a direct relationship to quantifying contextuality. 
In what follows, we will consider the set of ``free'' operations that include the ones introduced in \cite{Karol_Axiomatic} extended by the operations of wirings \cite{Cabello_wirings} and partial measurement.
We will show that the log-rank resource does not increase under such actions. 
The operations are:
\begin{enumerate}
    \item 
        \textbf{Adding a non-contextual behavior.} 
        By this operation we mean for a contextual challenging behavior $C(a \mid x)$ 
        \begin{equation}
            \mathcal{M}(C(a \mid x))  \mapsto C(a \mid x)\otimes P_{NC}(b \mid y).
        \end{equation}
    \item 
        \textbf{(Partial) measurement of a behavior}, e.g., 
        \begin{equation}
            P(a,b \mid x,y) \mapsto P(a,b \mid x=x_0,y),
        \end{equation}
        where a distribution $P(a,b \mid x=x_0,y)$ is still a behavior with possibly multiple input $y$ but single input $x=x_0$. 
        Measurement can be complete, resulting in behavior with single input and possibly multiple outputs, i.e., a distribution embedded into the polytope $\mathcal{P}$.
    \item 
        \textbf{Proper, simple wirings.} By simple wiring we mean here the operation:
        \begin{equation}
            P_1(a \mid x)\otimes P_2(b \mid y) \mapsto \sum_c P_1(a \mid c)\otimes P_2(c \mid y),
        \end{equation}
        where we assume that dimensions of output of $P_2$ (i.e., $|b|$), and the input of $P_1$ (i.e., $|x|$)  match. 
        We assume also it to be \textit{proper} which means that \textit{wiring is such, that it does not create contextuality when acting on an NC behavior}: for any $P_{NC}\in \mathcal{N}$ behavior, and proper wiring $W(P_{NC}) \in \mathcal{N}$. 
        For examples and the definition of such wiring see \cite{Karol_Axiomatic} and \cite{Cabello_wirings} respectively.
    \item 
        \textbf{Partial trace.} To give an example: 
        \begin{equation}
            P(a,b \mid x,y) \mapsto \sum_a P(a,b \mid x=x_0,y). 
        \end{equation}
        Let us note that this operation is well defined if and only if there is consistency condition satisfied (aka non-signaling), so that the marginal behavior does not depend on the input $x_0$. 
\end{enumerate}

We now prove the following:
\begin{theorem}
    Let the behavior $P$ defined as $P \equiv P(a_1,\ldots, a_{d_{\mathrm{out}}} \mid x_1,\ldots, x_{d_{\mathrm{in}}})$ be such that its marginals are well defined $\sum_{a \neq a_i} P(a \mid x_j = x_j^0 \mbox{ for }\,i\neq j,\, x_i)$. 
    Then, for any composition of operations $1$-$4$ defined above, call it $\Lambda$, we have that:
    \begin{equation}
        RC_2(\Lambda(P))\leq RC_2(P).
    \end{equation}
\end{theorem}
\begin{proof}
    The proof of the above theorem follows from the definition of the measure $RC_2$ and the closure properties of the set $\mathcal{N}$:
    \begin{enumerate}
        \item $\mathcal{N}$ is closed under the tensor product,
        \item $\mathcal{N}$ is closed under proper simple wirings,
        \item $\mathcal{N}$ is closed under partial trace.
    \end{enumerate}

    First, that $\mathcal{N}$ is closed under the tensor product is obvious -- the independent combination of two probability distributions over deterministic behaviors remains a probabilistic combination of deterministic behaviors. 
    The fact that $\mathcal{N}$ is closed under proper simple wirings follows from closure under the tensor product, and some more detailed arguments showed below. 
    Partial trace is similarly a restriction to certain deterministic behaviors.

    Let's fix a behavior $C(a \mid x)$ and represent an optimal set $S=\{NC^*_i\}_{i=1}^K$ of non-contextual behaviors realizing $RC_2(C)$ behaviors (this set need not be uniquely defined). 
    Then a set that simulates the challenging behavior $C\otimes P_{NC}$ is the following one: $\{NC^*_i\otimes P_{NC}\}_{i=1}^K$. 
    It does the job, as an input to the new behavior is a pair $(x, x')$ of inputs to the two of them. 
    Since behavior $P_{NS}$ simulates itself, it can process any input $x'$ (i.e., output of it under measurement $x$ is compatible with that of $P_{NC}(a \mid x)$) by definition. 
    And since $\mathcal{N}$ is closed under tensor product, this remains a non-contextual behavior. 
    Therefore, the tensor product does not increase the $RC_2$ measure.
    
    To see those simple proper wirings do not increase the rank of contextuality, we argue as follows. 
    Consider two behaviors of compatible input and outputs dimensions such that simple wiring is possible.
    \begin{align}
    	P(a \mid x)=\{P_{N_i}\}_{i=1}^{K_P}, \label{eq:sim-1}\\
    	P(b \mid y)=\{Q_{N_j}\}_{j=1}^{K_Q}, \label{eq:sim-2}
    \end{align}
    then consider a $\otimes$-simulation of the form: $P(a \mid x)\otimes P(b \mid y) \sim \{P_{N_i}\otimes Q_{N_j}\}_{i,j=1}^{K_P,K_Q}$. 
    Let us now see what happens to the simulation under simple wiring of the form 
    \begin{equation}
        P_{N_i}\otimes Q_{N_j} \mapsto \sum_{b'}P_{N_i}(a \mid b')\otimes P_{Q_j}(b' \mid y).
    \end{equation}
    
    We want to show (i) that the RHS equals $\sum_{b'}P(a \mid b')\otimes P(b' \mid y)$ and (ii) that the behaviors induced by wiring form a valid simulation, i.e., all belong to the set of non-contextual behaviors. 
    To see (i) let us first note that by Eq.\ (\ref{eq:sim-2}) for any $y$ there exists $Q_{N_{j=y}}$ such that for any $b$ there is  $P(b \mid y) =Q_{N_{j=y}}(b \mid y)$ hence:
    \begin{equation}
        \sum_{b'}P(a \mid b')\otimes P(b' \mid y)=\sum_{b'}P(a \mid b')\otimes Q_{N_{j=y}}(b' \mid y).
    \end{equation}
    
    Now by Eq.\ (\ref{eq:sim-1}), for any $x$, e.g., $x=b'$ there exists $P_{N_{i=x}}$ such that for any $a$ $P(a \mid b') =P_{N_{i=x}}(a \mid b')$ thus, 
    \begin{equation}
        \mathop\forall_{a,b',y} P(a \mid b')\otimes P(b' \mid y) =
        P_{N_{i=b'}}(a \mid b')\otimes Q_{N_{j=y}}(b' \mid y),
    \end{equation}
    hence,
    \begin{equation}
        \begin{split}
            \mathop\forall_{a,y}& \sum_{b'} P(a \mid b')\otimes P(b' \mid y)\\
            &=\sum_{b'}  P_{N_{i=b'}}(a \mid b')\otimes Q_{N_{j=y}}(b' \mid y)\\
            &\equiv P'_{N}(a \mid y)_{(ij)}.
        \end{split}
    \end{equation}
    
    The set of non-contextual behaviors is closed under wirings, hence the behavior $P_{N_i}\otimes Q_{N_j}$ is non-contextual for any $i,j$.
    Since the wiring is proper, it sends non-contextual behaviors to non-contextual ones. 
    Because of that the behavior $P'_{N}(a \mid y)_{(ij)}$ is non-contextual.
    Since this holds for any $(i,j)$ which pairs is $K_P\times K_Q$ many, we have that 
    \begin{equation}
        \begin{split}
            W(a \mid y) &\equiv\sum_{b'} P(a \mid b')\otimes P(b' \mid y)\\
            &\sim \left\{P'_{N}(a \mid y)_{(ij)}\right\}_{i,j=1}^{K_P,K_Q}.
        \end{split}
    \end{equation}
    
    Since this is just one possible simulation, and the rank of contextuality is the infimum over cardinalities of the latter, we get that:
    \begin{equation}
        \begin{split}
            RC_2(W(a \mid y)) &\leq K_P\times K_Q\\
            &= RC_2(P(a \mid x)\otimes P(b \mid y)).
        \end{split}
    \end{equation}
    Finally, we will show that (partial) measurement of a behavior do not increase the log-rank.
    It is enough to show that
    \begin{equation}
        \begin{split}
            \mathop\forall_{\substack{a=a_0\\x=x_0}} RC & \left(\frac{P(a=a_0, b \mid x=x_0, y)}{P(a=a_0 \mid x=x_0)} \right)\\
            & \leq RC(P(a, b \mid x, y)).
        \end{split}
    \end{equation}
    From the definition of the rank of contextuality, we know that behavior on the RHS can be simulated by some set of non-contextual behaviors $S = \{ NC_i(a, b \mid x, y) \}_{i=1}^K$. Let 
    \begin{equation}
        NC_i^{a_0|x_0}(b \mid y) \coloneqq \frac{NC_i(a=a_0, b \mid x=x_0, y)}{NC_i(a=a_0 \mid x=x_0)}.
    \end{equation}
    We can now define the set 
    \begin{equation}
        S' \coloneqq \left\{ {NC_i^{a_0|x_0}(b \mid y)} \right\}_{i=1}^{K'},
    \end{equation}
    (where $K'\leq K$ as some marginals of $NC_i(a,b \mid x,y)$ given input $x_0$ and output $a_0$ can coincide).
    In this set every behavior $NC_i^{a_0|x_0}$ is result of measurement observable $x=x_0$ on behavior $NC_i$ that gives output $a=a_i$.
    Since all behaviors $NC_i$ were non-contextual, we have that all $NC_i^{a_0|x_0}$ behaviors are also non-contextual.
    We also know that $|S'| \leq |S|$.
    The behavior from LHS can be simulated by the set $S'$, although it is not necessarily true that the $S'$ is the smallest possible simulating set.
    We get, therefore, LHS $\leq$ RHS. 
    The analogous inequality is also true for $RC_2$, which follows from the monotonicity of the logarithm that completes the proof.
\end{proof}

\section{Additivity of Log-Rank}
\label{sec:aditivity}

In addition to the structure of a resource theory, the log-rank measure of contextuality admits additivity under the tensor product operation.
\begin{theorem}
    For any behaviors $P_i \in \mathcal{C}$, 
    \begin{equation}
        RC_2\left(\mathop\bigotimes_{i=1}^n P_i\right)=\sum_{i=1}^n RC_2(P_i).
    \end{equation}
\end{theorem}

Before the proof, we will present a short lemma that follows from the definition of $RC_2$:
\begin{lemma}
	Let $\{NC_{i}(a, b \mid x,y)\}_{i=1}^{K}$ is such that $K=RC_2(P(ab \mid xy))$ and $P(ab \mid xy)=P_{1} (a \mid x) \otimes P_{2}(b \mid y)$ where $P_{1}$ and $P_{2}$ are two arbitrary behaviors. 
	Then we have that the family of non-contextual behaviors
	\begin{equation}
	    \{\mathrm{Tr}_{2}(NC_{i}(a, b \mid x,y))\otimes \mathrm{Tr}_{1}(NC_{i}(a, b \mid x,y))\}_{i=1}^{K}
	\end{equation}
	is a non-contextual cover of $P(ab \mid xy)=P_{1} (a \mid x) \otimes P_{2}(b \mid y)$. 
\end{lemma}
\begin{proof}
	By assumption, for all $x,y$ context pairs, there exists an $i$ such that $ NC_{i}(a,b \mid x,y) = P_{1} (a \mid x) \otimes P_{2}(b \mid y)$. 
	Then for that input pair, by definition:
	\begin{align}
	    \mathrm{Tr}_{1} NC_{i}(a,b \mid x,y) &= P_{2}(b \mid y),\\
	    \mathrm{Tr}_{2} NC_{i}(a,b \mid x,y) &=P_{1}(a \mid x).
	\end{align}
	It demonstrates that for all $x,y$, the product of the traces suffices for the simulation. 
\end{proof}
\begin{proof} 
	The proof of the theorem involves first proving the basic case of two contextual behaviors $P_{1}\otimes P_{2}$. 
	First, given $RC_2(P_{1})= \log_{2} K_{1}$ and $RC_2(P_{2})=\log_{2}K_{2}$, a naive simulation is simply to take the tensor product of all non-contextual behaviors in each minimal simulation.
	So it is clear that $RC_2(P_{1} \otimes P_{2}) \leq \log_{2}(K_{1}\cdot  K_{2}) = \log_2 K_1 + \log_2 K_2 =RC_2(P_{1}) + RC_2(P_{2})$. 
	
	By the Lemma, a given simulation of $P_1 \otimes P_2$ can be written as $\{N_{i}(a \mid x) \otimes N_{i}(b \mid y)\}_{i=1}^{K}$. 
	By definition $\{\mathrm{Tr}_{1}N_{i}(a \mid x) \otimes N_{i}(b \mid y)\} \sim P_{2}$, so ${K} \geq K_{2}$. 
	
	The set of inputs $y$ to $P_2$ are then divided into at least $K_2$ subsets such that each subset of inputs requires a different (one of at least $K_2$) non-contextual behavior $N_i(b \mid y)$ according to definition $RC(P_2)$. 
	Let us call the set of representants of these subsets as $\hat{\mathcal{Y}}$. We have then $|\hat{\mathcal{Y}}| \geq K_2$.
	
	Let us then fix $y_{0} \in \hat{\mathcal{Y}}$ arbitrarily, such that the simulation of $P_2$ requires a non-contextual behavior $\hat{N}(b \mid y)$ when measurement $y_0$ is chosen and there exists some $i$ such that the marginal over system $1$ of $N_i(a \mid x)\otimes N_i(b \mid y)$ equals $\hat{N}(b \mid y)$ (it follows from the above that such $i$ exists). 
	We then ask how many behaviors of the form $N_j(a \mid x)\otimes N_j(b \mid y)$ has $N_j(b \mid y)=\hat{N}(b \mid y)$. 
	Let us denote the set of indices $j$ satisfying this property $\mathcal{J}$. 
	To answer this question observe that the family $\{ N_{j}(a \mid x)\otimes\sum_{b}  N_{j}(b \mid y_{0})\}_{j \in \mathcal{J}}$ simulates $P_{1}$. 
	Indeed, the set which simulates $P_1\otimes P_2$ must be ready for every pair of inputs $(x,y_0)$, and partial trace over system $2$ of non-contextual behavior which covers input $(x,y_0)$ is also non-contextual and covers input $x$.
	We will argue now that $|\mathcal{J}| \geq K_1$.
	Otherwise we could find a set of non-contextual behaviors of the form $\{ N_{j}(a \mid x)\otimes\sum_b N_j(b \mid y_0)\}_{j\in \mathcal{J}}$ that simulates $P_1$ with less number of elements than $K_1$. 
	This, however, would contradict the fact that $RC(P_1)=K_1$. 
    Since $y_0 \in \hat{\mathcal{Y}}$ was arbitrary, the same argument goes for any $y \in \hat{\mathcal{Y}}$, of which there are $K_2$ many. 
    This implies that $K \geq K_1 \times K_2$, which completes the proof in the two-behavior case.	

    The general case follows from induction and the fact that the tensor product of contextual behaviors is itself contextual.
\end{proof}

It turns out that the rank is not convex.
However, it is upper bounded by the maximal value of the ranks of the mixed behaviors, which we state below.
\begin{remark}
    Let $P_1$ and $P_2$ are two behaviors defined on the same graph $G$. 
    Then for any $\lambda \in \{ 0, 1 \} $ the mixture of probabilities of those behavior $P \coloneqq \lambda P_1 + (1 - \lambda) P_2$ fulfils
    \begin{equation}
        RC_2(P) \leq \max \{ RC_2(P_1), RC_2(P_2) \}.
    \end{equation}
\end{remark}

\section{Constructions of behaviors with arbitrary Log-Rank}
\label{sec:unboundedness}

In this section, we will show how to construct a behavior that has an arbitrary log-rank.
We will restrict our considerations here only to the case of standard graphs. 
It means that we consider the case when, in a hypergraph, all the edges contain only two vertices, so the hypergraph becomes the usual (undirected) graph.
Before that, we will present some toy examples of calculating log-rank in simple cases.

We will start with an example of a 3-cycle behavior and then extend it to a cycle of arbitrary length.
The 3-cycle behavior consists of 3 observables $\{ A, B, C \}$ which has output $O = \{ 0, 1 \}$ and 3 contexts $C_1 = \{ A, B \}$, $C_2 = \{ B, C \}$, and $C_3 = \{ C, A \}$, and each of conditional probability is anticorrelation
\begin{equation}
    \label{eq:probability3cycle}
    \begin{cases}
        P(01 \mid C_i) = P(10 \mid C_i) = \frac{1}{2},\\
        P(00 \mid C_i) = P(11 \mid C_i) = 0,
    \end{cases} 
\end{equation}
for every $C_i$.
The 3-cycle behavior is contextual since no single non-contextual behavior reproduces all conditional probabilities. 
On the other hand $\{ P(o \mid C_i) \}_{i=1,2}$ and $\{ P(o \mid C_3) \}$ is extendable respectively.
For example $\{ P(o \mid C_i) \}_{i=1,2}$  can be reproduced by a non-contextual behavior $N_1$:
\begin{equation}
    \label{repair1}
    P_{N_1}(010 \mid ABC) = P_{N_1}(101 \mid ABC) = \frac{1}{2}
\end{equation}
and $\{ P(o \mid C_3) \}$ can be reproduced by a non-contextual behavior $N_2$:
\begin{equation}
    \label{repair2}
    P_{N_2}(001 \mid ABC) = P_{N_2}(100 \mid ABC) = \frac{1}{2}.
\end{equation}
These two non-contextual behaviors $N_1, N_2$ simulate 3-cycle behavior. 
Therefore $RC(3\text{-cycle}) \leq 2$.

We can easily extend this example to the case of arbitrary $k$-cycle behavior.
In the case when $k$ is odd, we assign anticorrelations to all of the edges the same way as in $3$-cycle one (see Eq.\ (\ref{eq:probability3cycle})). 
On the other hand, if $k$ is even, this will not work. 
The simplest solution is to assign correlation (identity operation) to one edge and assign anticorrelations to all remaining $k-1$ edges. 
The probability distribution would be then
\begin{equation}
    \begin{cases}
        P(01 \mid C_i) = P(10 \mid C_i) &= \frac{1}{2} \quad \text{ for } i > 1,\\
        P(00 \mid C_i) = P(11 \mid C_i) &= 0 \quad \text{ for } i > 1,\\
        P(00 \mid C_1) = P(11 \mid C_1) &= \frac{1}{2},\\
        P(01 \mid C_1) = P(10 \mid C_1) &= 0.
    \end{cases}
\end{equation}
Therefore $RC(k\text{-cycle}) \leq 2$.

\begin{figure}[htb]
    \subfloat[Bipartite graph $K_{3,3}$\label{fig:graph1}]{%
        \begin{tikzpicture}

  \filldraw[black] (0,0) circle (2pt);
  \filldraw[black] (0,1) circle (2pt);
  \filldraw[black] (0,2) circle (2pt);

  \filldraw[black] (2,0) circle (2pt) ;
  \filldraw[black] (2,1) circle (2pt) ;
  \filldraw[black] (2,2) circle (2pt) ;

  \draw (0,0) -- (2,0);
  \draw (0,0) -- (2,1);
  \draw (0,0) -- (2,2);
  \draw (0,1) -- (2,0);
  \draw (0,1) -- (2,1);
  \draw (0,1) -- (2,2);
  \draw (0,2) -- (2,0);
  \draw (0,2) -- (2,1);
  \draw (0,2) -- (2,2);

\end{tikzpicture}%
    }
    \hfill
    \subfloat[Wheel graph $W_5$\label{fig:graph2}. 
        Solid and dashed lines shows division into two forests.]{%
        \begin{tikzpicture}

  \filldraw[black] (0,0) circle (2pt);
  \filldraw[black] (0,2) circle (2pt);

  \filldraw[black] (2,0) circle (2pt) ;
  \filldraw[black] (2,2) circle (2pt) ;
  
  \filldraw[black] (1,1) circle (2pt) ;

  \draw [dashed] (0,0) -- (2,0);
  \draw [dashed] (2,0) -- (2,2);
  \draw [dashed] (2,2) -- (0,2);
  \draw (0,2) -- (0,0);
  
  \draw (1,1) -- (0,0);
  \draw [dashed] (1,1) -- (0,2);
  \draw (1,1) -- (2,0);
  \draw (1,1) -- (2,2);
  
\end{tikzpicture}%
    }
    \hfill
    \subfloat[Complete graph $K_5$\label{fig:graph3}]{%
        \begin{tikzpicture}

  \filldraw[black] (0.309,0.951) circle (2pt);
  \filldraw[black] (1,0) circle (2pt);
  \filldraw[black] (-0.809,0.587) circle (2pt);
  \filldraw[black] (-0.809,-0.587) circle (2pt) ;
  \filldraw[black] (0.309,-0.951) circle (2pt) ;

  \draw (0.309,0.951) -- (1,0);
  \draw (0.309,0.951) -- (-0.809,0.587);
  \draw (0.309,0.951) -- (-0.809,-0.587);
  \draw (0.309,0.951) -- (0.309,-0.951);
  \draw (1,0) -- (-0.809,0.587);
  \draw (1,0) -- (-0.809,-0.587);
  \draw (1,0) -- (0.309,-0.951);
  \draw (-0.809,0.587) -- (-0.809,-0.587);
  \draw (-0.809,0.587) -- (0.309,-0.951);
  \draw (-0.809,-0.587) -- (0.309,-0.951);

\end{tikzpicture}%
    }
    \hfill
    \subfloat[Cycle graph $C_5$\label{fig:graph4}]{%
        \begin{tikzpicture}

  \filldraw[black] (1,0) circle (2pt);
  \filldraw[black] (0.309,0.951) circle (2pt);
  \filldraw[black] (-0.809,0.587) circle (2pt);
  \filldraw[black] (-0.809,-0.587) circle (2pt) ;
  \filldraw[black] (0.309,-0.951) circle (2pt) ;

  \draw (1,0) -- (0.309,0.951);
  \draw (0.309,0.951) -- (-0.809,0.587);
  \draw (-0.809,0.587) -- (-0.809,-0.587);
  \draw (-0.809,-0.587) -- (0.309,-0.951);
  \draw (0.309,-0.951) -- (1,0);

\end{tikzpicture}%
    }
    \caption{
        \label{fig:graphs}
        Examples of different graph types used in the text.
    }
\end{figure}
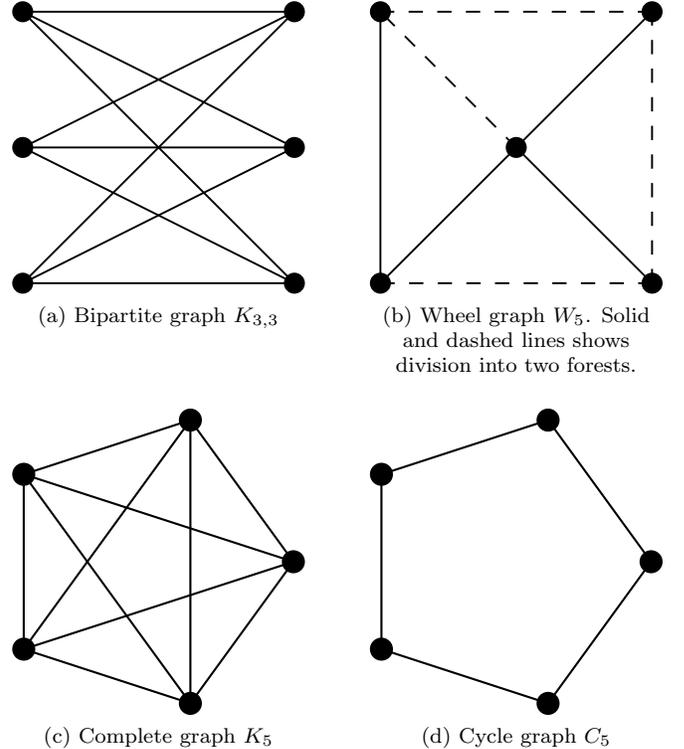

Now, we will present some definitions and lemmas needed to show how to construct a behavior with an arbitrarily high value of log-rank.
\begin{definition}[Arboricity]
    For a graph $G$, the arboricity $\Upsilon(G)$ is the minimum number of spanning forests (edge-disjoint acyclic subgraphs) whose union is $G$.
\end{definition}
\begin{definition}[Cycle graph]
    A cycle graph $C_{n}$ is a graph that has $n$ vertices and $n$ edges. 
    Such a graph consists of a single cycle that contains all vertices.
    See Figure \ref{fig:graph4} for an example of cycle graph.
\end{definition}
\begin{definition}[Complete graph]
    A complete graph $K_{n}$ is a graph that have $n$ vertices and $n(n-1)/2$ vertices. 
    In a complete graph, every two vertices are connected by an edge.
    See Figure \ref{fig:graph3} for an example of complete graph.
\end{definition}
\begin{lemma}
    Let $K_n$ be complete graph. 
    Then the arboricity of the graph equals $\Upsilon(K_n) = \lceil n/2 \rceil$.
\end{lemma}
\begin{definition}[Complete bipartite graph]
    A complete bipartite graph $K_{m,n}$ is a graph that have $m+n$ vertices that are partitioned into two sets $V_1$ and $V_2$ of size $m$ and $n$ respectively. 
    The graph have $mn$ edges distributed in such a way that two vertices $v_1, v_2 $ are connected by an edge if and only if $v_1 \in V_1$ and $v_2 \in V_2$.
    See Figure \ref{fig:graph1} for an example of complete bipartiate graph.
\end{definition}
\begin{lemma}
    \label{lem:bip}
    Let $K_{m,n}$ be a complete bipartite graph; that is a graph in which we can partition vertices into two sets of vertices $V_1$. 
    Then the arboricity of the graph equals 
    \begin{equation}
        \Upsilon(K_{m,n}) = \left\lceil \frac{mn}{m+n-1} \right\rceil.
    \end{equation}
\end{lemma}
\begin{definition}[Wheel graph]
    A wheel graph $K_{n}$ is a graph that has $n$ vertices and $2(n-1)$ edges. 
    The graph is formed by connecting a single universal vertex to all vertices of a cycle $C_{n-1}$.
    See Figure \ref{fig:graph2} for an example of wheel graph.
\end{definition}
For more information about arboricity, see, for example, Graph Theory book by Frank Harary \cite{graph}.

From the toy example, we know that a behavior represented by an acyclic graph is non-contextual. 
Additionally, the one represented by a cycle can have the rank of contextuality 0 or 1 depending on the probabilities. 
We will now show that our measure of contextuality is upper bounded by the arboricity and how, for every graph, construct a contextual behavior that has the rank of contextuality equal to its arboricity.
It is easy to see that rank of contextuality can not be greater than arboricity since we can divide the graph $G$ that represents the behavior into $\Upsilon(G)$ forests that are always non-contextual.
It is, however, more complicated to show, that for every graph, we can assign probabilities in such a way that we obtain the maximal rank of contextuality.
We should emphasize that the simple method we used for cycles will not always work since many different cycles in the graph can share the same edge.
For example, in the wheel graph $W_5$ (See Figure \ref{fig:graph2}), it is impossible to assign correlations and anticorrelations to the edges in such a way that there is a contradiction for every cycle. 
Therefore, have to use more elaborate construction.
\begin{definition}[color]
    For any parameter $m \in \mathbb{N}$ let us define set of permutations $\{\tau_o\}_{i=1}^m$ acting on the set $V = \{ 0, 1, \ldots, 2^m-1 \}$ in the following way
    \begin{equation}
        \tau_i(v) \coloneqq v + (-1)^{\left(\left\lfloor \frac{v}{2^{i-1}} \right\rfloor \mod 2\right)} 2^{i-1}.
    \end{equation}
    We will call such permutations colors.
\end{definition}
In Figure \ref{fig:colors} we present an example of colors.

\begin{figure}[htbp]
    \begin{tikzpicture}
  \draw[red] (0,0) -- (2,1);
  \draw[red] (0,1) -- (2,0);
  \draw[red] (0,2) -- (2,3);
  \draw[red] (0,3) -- (2,2);
  \draw[red] (0,4) -- (2,5);
  \draw[red] (0,5) -- (2,4);
  \draw[red] (0,6) -- (2,7);
  \draw[red] (0,7) -- (2,6);
  
  \draw[green, dotted, thick] (2,0) -- (4,2);
  \draw[green, dotted, thick] (2,2) -- (4,0);
  \draw[green, dotted, thick] (2,1) -- (4,3);
  \draw[green, dotted, thick] (2,3) -- (4,1);
  \draw[green, dotted, thick] (2,4) -- (4,6);
  \draw[green, dotted, thick] (2,6) -- (4,4);
  \draw[green, dotted, thick] (2,5) -- (4,7);
  \draw[green, dotted, thick] (2,7) -- (4,5);
  
  \draw[blue, dashed] (4,0) -- (6,4);
  \draw[blue, dashed] (4,4) -- (6,0);
  \draw[blue, dashed] (4,5) -- (6,1);
  \draw[blue, dashed] (4,1) -- (6,5);
  \draw[blue, dashed] (4,6) -- (6,2);
  \draw[blue, dashed] (4,2) -- (6,6);
  \draw[blue, dashed] (4,7) -- (6,3);
  \draw[blue, dashed] (4,3) -- (6,7);
  
  \filldraw[black] (0,0) circle (1.5pt) node[anchor=east]{7};
  \filldraw[black] (0,1) circle (1.5pt) node[anchor=east]{6};
  \filldraw[black] (0,2) circle (1.5pt) node[anchor=east]{5};
  \filldraw[black] (0,3) circle (1.5pt) node[anchor=east]{4};
  \filldraw[black] (0,4) circle (1.5pt) node[anchor=east]{3};
  \filldraw[black] (0,5) circle (1.5pt) node[anchor=east]{2};
  \filldraw[black] (0,6) circle (1.5pt) node[anchor=east]{1};
  \filldraw[black] (0,7) circle (1.5pt) node[anchor=east]{0};
  
  \filldraw[black] (2,0) circle (1.5pt) ;
  \filldraw[black] (2,1) circle (1.5pt) ;
  \filldraw[black] (2,2) circle (1.5pt) ;
  \filldraw[black] (2,3) circle (1.5pt) ;
  \filldraw[black] (2,4) circle (1.5pt) ;
  \filldraw[black] (2,5) circle (1.5pt) ;
  \filldraw[black] (2,6) circle (1.5pt) ;
  \filldraw[black] (2,7) circle (1.5pt) ;
  
  \filldraw[black] (4,0) circle (1.5pt) ;
  \filldraw[black] (4,1) circle (1.5pt) ;
  \filldraw[black] (4,2) circle (1.5pt) ;
  \filldraw[black] (4,3) circle (1.5pt) ;
  \filldraw[black] (4,4) circle (1.5pt) ;
  \filldraw[black] (4,5) circle (1.5pt) ;
  \filldraw[black] (4,6) circle (1.5pt) ;
  \filldraw[black] (4,7) circle (1.5pt) ;
  
  \filldraw[black] (6,0) circle (1.5pt) node[anchor=west]{7};
  \filldraw[black] (6,1) circle (1.5pt) node[anchor=west]{6};
  \filldraw[black] (6,2) circle (1.5pt) node[anchor=west]{5};
  \filldraw[black] (6,3) circle (1.5pt) node[anchor=west]{4};
  \filldraw[black] (6,4) circle (1.5pt) node[anchor=west]{3};
  \filldraw[black] (6,5) circle (1.5pt) node[anchor=west]{2};
  \filldraw[black] (6,6) circle (1.5pt) node[anchor=west]{1};
  \filldraw[black] (6,7) circle (1.5pt) node[anchor=west]{0};
  
  \node at (1,7.3) {\textcolor{red}{$\tau_1$}};
  \node at (3,7.3) {\textcolor{green}{$\tau_2$}};
  \node at (5,7.3) {\textcolor{blue}{$\tau_3$}};

\end{tikzpicture}
    \caption{
        \label{fig:colors}
        Example of the set of three permutation (colors) acting on eight numbers.
    }
\end{figure}
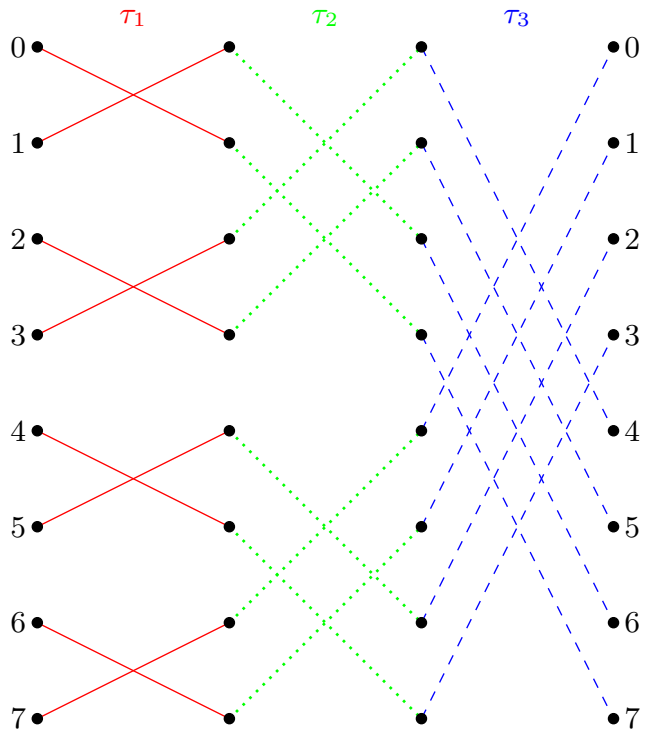

Although the above definition may appear complex, the reader should easily understand it in terms of bit operations and binary representations. 
For every number $v \in V$ we can assign its binary representation $v = b_m b_{m-1} \ldots b_i \ldots b_2 b_1$ where $b_i \in \{0,1\}$. 
Then the permutation $\tau_i$ is just bit flip on i-th position of binary representation. 
We can formally write it as
\begin{equation}
    \begin{split}
        \tau_i(v) &= \tau_i(b_m b_{m-1} \ldots b_i \ldots b_2 b_1)\\
        &= b_m b_{m-1} \ldots \bar{b_i} \ldots, b_2 b_1
    \end{split}
\end{equation}
where $\bar{b_i}$ denotes bit flip.

Note that our definition of colors is different from the one presented in Eq.\ (3) and Eq.\ (4) of the paper by Rosicka et al.\ \cite{Rosicka2016}). 
Although, we will show that our colors also fulfill properties P1 and P2 from \cite{Rosicka2016}.
\begin{itemize}
    \item (P1) Each permutation is symmetric with respect to the exchange of players, i.e., the permutations are their own inverse.
    \item (P2) Every pair $(v, \tau(v))$ appears exactly once in the set of permutations (in particular, each permutation assigns a different $\tau(v)$ for each given $v \in V$).
\end{itemize}

It is easy to see that symmetry comes from the fact that using the same permutation twice (which means flipping the same bit twice) gives us identity. 
The second property is also obvious since any number of different permutations (flipping two different bits of the same number) always give different results. 
\begin{lemma}
    \label{lem:colors}
    The permutation defined by the composition of any number of different colors does not have fixed points.
\end{lemma}
\begin{proof}
    Let $\tau$ be composition of any number of different colors meaning 
    \begin{equation}
        \tau \coloneqq \tau_{i_1} \circ \tau_{i_2} \circ \cdots \circ \tau_{i_k}  
    \end{equation}
    where all $i_j$ are different numbers from 0 to $m$.
    We have to show that for all $v \in V$ we have $\tau(v) \neq v$.
    We know, from binary representation, that $\tau_i$ is bit flip on position $i$ then $\tau_i$ is just operation that flips bits on all positions $i_j$. 
    Since there is at least one permutation in $\tau$ and no permutation appears more than once, at least one bit is always different for any value $v$. 
    Therefore $\tau$ does not have fixed points.
\end{proof}
\begin{theorem}
    \label{thm:arb}
    For every graph $G$, there exists a behavior with the rank of contextuality equal to the arboricity of the graph.
\end{theorem}
\begin{proof}
    If the graph is acyclic (arboricity equals zero by definition), it is always non-contextual.
    Let us assume that the graph $G$ has $m$ edges and arboricity $\Upsilon(G)$. 
    We will assign a different color to each edge.
    From the definition of arboricity, we know that we can not divide it into less than $\Upsilon(G)$ graph without cycles. 
    Let us assume that the rank of contextuality is less than $\Upsilon(G)$.
    Then at least one of the non-contextual behaviors simulating the given behavior is represented by a subgraph that contains a cycle.
    Since all edges have different colors, also this cycle edges have different colors.
    We will then assign probabilities to each edge according to its color (permutation) using the formula
    \begin{equation}
        \label{Contextforest}
        \begin{cases}
            P(o_1 o_2 \mid C_i) = \frac{1}{2^m} &\text{ if } o_2 = \tau_i (o_1),\\
            P(o_1 o_2 \mid C_i) = 0 &\text{ otherwise}.
        \end{cases}
    \end{equation}
    From Lemma \ref{lem:colors} we know that permutations represented by the composition of different colors do not have fixed points.
    Therefore, there does not exist any deterministic assignment of values for the non-contextual behavior simulating the given behavior which leads to a contradiction.
\end{proof}

\section{Comparison to Other Contextuality Measures}
\label{sec:comparison}

Here we show that log-rank yields different results than commonly used measures for specific behaviors, justifying that it is a novel measure of contextuality. 
We first introduce the contextual fraction \cite{Abramsky_contextualfrac}, then the analogous robustness of contextuality \cite{Amaral_Geo}, followed by relative entropy of contextuality \cite{Grudka_contextuality}, and finally the contradiction number \cite{Ramanathan2014}. 

\begin{definition}[Contextual fraction]
    \label{def:measure1}
    The definition of the contextual fraction is given by the formula
    \begin{equation}
        \mathcal{CF}(P) \coloneqq  \min \left\{ \lambda : P= \lambda P'+(1-\lambda)P_{NC} \right\}.
    \end{equation}
\end{definition}

\begin{definition}[Robustness of contextuality]
    \label{def:measure2}
    The robustness of contextuality is defined as
    \begin{equation}
        \mathcal{R}(P) \coloneqq \min \left\{ \lambda : (1- \lambda) P+\lambda P_{NC} \in \mathcal{C} \right\}.
    \end{equation}
\end{definition}

\begin{definition}[(Uniform) relative entropy of contextuality]
    \label{def:measure3}
    The definition of the relative entropy of contextuality is given by
    \begin{equation}
        \begin{split}
           X_{\mathrm{max}}(P) & \coloneqq\\
           \sup_{P(X)} & \min_{NC(A|X)} \sum_{X \in E} P(X) D(P(A \mid X) \;||\; NC(A \mid X)) 
        \end{split}
    \end{equation}
    where $D$ is the relative entropy distance (for definition see for example \cite{cover1991elements}).
    We also define different variant of above quantity called uniform relative entropy of contextuality as:
    \begin{equation}
        X_{u}(P) \coloneqq \min_{NC(A|X)} \sum_{X \in E} \frac{1}{|E|} D(P(A \mid X) \;||\; NC(A \mid X)).
    \end{equation}
\end{definition}

\begin{definition}[Contradiction number]
    \label{def:measure4}
    The contradiction number of behavior $B$ is the minimal number of observables (vertices of the graph) that have to be removed from $B$ in order to obtain non-contextual behavior.
\end{definition}

\begin{lemma}
    \label{lem:compare}
    The log-rank measure of contextuality differs from other contextuality measures presented in the Definitions \ref{def:measure1}, \ref{def:measure2}, \ref{def:measure3}, \ref{def:measure4} and the memory of contextuality of \cite{Kleinmann_memory, Fagundes}.
\end{lemma}
\begin{proof}
    Here we provide a simple example of the convex combination of the contextual uniform Mermin-square behavior ($MS$) and isotropic behavior of equally-weighted outcomes ($U$), i.e., 
    \begin{equation}
        C = \lambda MS +(1- \lambda)U.
    \end{equation}

    Recall from Figure \ref{fig:PM} that $RC_2(PM)=1$, with two uniform behaviors being sufficient for simulation. 
    Each individual behavior is sufficient to simulate the isotropic behavior. 
    Then the number of behaviors becomes $2^{2}$ to account for the convex combination selection weighting. 

    It is evident from the fact that Mermin's square has contextual fraction $1$, that $\mathcal{CF}(C)=\lambda$. 
    Similarly, $\mathcal{R}(C)=1-\lambda$ with $\lambda\in [0,1]$. 
    
    To show the difference from (uniform) relative entropy of contextuality, we will use an example of isotropic Popescu-Rohrlich box denoted as $PR_{\alpha}$ with $\alpha > 3/4$.
    It is easy to see that $RC(PR_{\alpha})=2$ and $RC_2(PR_{\alpha})=1$.
    On the other hand it was shown in \cite{Grudka_contextuality} that 
    \begin{equation}
        X_{\mathrm{max}}(PR_{\alpha}) = X_{u}(PR_{\alpha}) = \log \left( \frac{4}{3^\alpha} \right) + h(\alpha)
    \end{equation}
    where $h(\alpha) \coloneqq -\alpha \log \alpha (1 - \alpha) \log (1 - \alpha)$.
    Therefore, since the values are clearly different, (uniform) relative entropy of contextuality is different from our new measures.
    Concluding, in all of these cases, the log-rank differs from these previously explored measures.
    
    We compare it now with measures that take integer values. 
    We first note that our measure differs from the memory of contextuality for the Pers-Mermin set of observables, as it equals $3$ \cite{Kleinmann_memory, Fagundes}, while the value of our measure equals $2$ (see Figure \ref{fig:aut}).
    Regarding the number of contradictions, the latter takes value $1$ for $n$-cycle graph, and $2$ for the wheel graph  $W_n$ (see Figure \ref{fig:graphs}), while the arboricity (which is equal to the rank of contextuality for specially constructed behaviors), is equal to $2$ for both the graphs (see Figure. \ref{fig:graphs}).
    An example showing a more significant difference than $1$ comes from Theorem \ref{thm:arb}. 
    It is shown that for the bipartite graph $K_{m,n}$ there exists a behavior with the $RC$ equal to $\lceil (mn)/(m+n-1) \rceil$ (see Lemma \ref{lem:bip}). 
    On the other hand, it is easy to check that the contradiction number equals $\min \{n,m\}-1$. 
    It is easy to see that if there are left $2$ observables in both bipartitions, then a cycle (of length $4$) is left in a graph. 
    This cycle by construction can be contextual.
    Hence in one bipartition, there must be left only $1$ vertex, and the other bipartition can be left untouched.
    It is then clear that the measures differ for $K_{n,m}$ by about twice.
\end{proof}

\section{Towards potential applications}
\label{sec:applications}

In this section, we show possible applications of our measure. 
We will do it regarding two unrelated phenomena.  
In the first one, we propose a benchmark useful for data management of probabilistic databases inspired by the simulating automaton (\ref{automaton}) and the rank of contextuality.
The second one concerns the ways of attack by the so-called non-signaling adversary. 
In the second example, we focus on the Bell non-locality scenario -- a specific case of the contextuality in which commutation relations between involved observables are enforced by the no faster than light condition \cite{Bell-nonlocality}. 
More precisely, the corresponding hypergraphs are bipartite graphs $K_{m,n}$ where one party has $m$ and the other $n$ observables.

\subsection{Management of unclean probabilistic databases}

The article \cite{Abramsky2013} introduced a description of contextuality in the language of relational databases (RDBs). 
This connection resulted in many contributions to the empirical description of contextuality  \cite{Abramsky_contextualfrac}, the discovery of hierarchies between the contextuality demonstrations \cite{Abramsky2011, AbramRui2018} and the underlying logical models \cite{kishida2016}.
While such a connection has shown the value of importing data science methods into the study of quantum contextuality, we propose a complementary approach: The export of methods and results from contextuality to data management.

Before presenting our proposal, it will be instructive to define some fundamental concepts of data management in the language introduced in our work.

A \textit{database} consists of a collection of tables, in which the first row presents the properties of one or more entities and the subsequent rows show the possible simultaneous values of such properties.
Data science nomenclature for tables and properties is \textit{relations} and \textit{attributes}, respectively \cite{page1990}. 
In quantum information, the entities studied are always physical systems whose attributes identify with observables whose values an agent could tabulate in a context.

Reversing the previous assignment, we can imagine a database as a collection of an entity's observables $A_{1}, A_{2}, \ldots, A_{|X|}$ tabulated in the first row and each possible output $\mathbf{a} = a_{1}, a_{2}, \ldots, a_{|X|}$ in subsequent rows, with a different table for each context $X$.
The possible outputs $\mathbf{a}$  are denoted as \textit{tuples} over the corresponding table (relation) with attributes $\{X\}$ \cite{page1990}. 
In a database, tuples usually provide the data stored about each entity studied.

So far, the above would suffice to describe in the language of contextuality the fundamental concepts of relational databases and recover the connection described in \cite{Abramsky2013}. 
However, we can further extend our description to encompass a field that emerges at the intersection of machine learning and database theory: Probabilistic databases \cite{Sa2019, Suciu2020}.

Data management systems such as NELL \cite{Nell2018}, DeepDive \cite{Deep2015}, and YAGO \cite{Yago2013} continuously crawl the Web to extract structured information. 
Also, private Projects such as Microsoft's Probase \cite{Microsoft2012} or Google's Knowledge Vault \cite{Google2014} similarly learn structured data from text and fill their databases with millions of entities and billions of facts. 
The above information extraction systems employ statistical machine learning techniques to produce a probabilistic prediction. 
Therefore, it is common to interpret such~large-scale knowledge bases~through probabilistic semantics. 

The standard framework to represent probabilistic data is precisely that of probabilistic databases (PDBs).
In the language of our work a PDB is simply a database where in each table $X$ we associate a probability $p_{i}$ and a deterministic distribution $ T_{i}\left(\mathbf{a}=\mathbf{a}_i\mid X\right)$ to every possible tuple $\mathbf{a}_i$, which is one if $\mathbf{a}=\mathbf{a}_i$ and zero otherwise. 
Asking for the probability atributes $A_{1}, A_{2}, \ldots, A_{\left|X\right|}$ to have the values $\mathbf{a}$ in table $X$ is a particular kind of \textit{query} in database language. 
The answer to the previous query in a PDB represented by $\left\{ T_{i}\left(\mathbf{a}=\mathbf{a}_i\mid X\right),p_{i}\right\}$ is given by a probability $P\left(\mathbf{a}\mid X\right)$:
\begin{equation}
    \label{Probqueryanswer}
    P(\mathbf{a} \mid X) = \sum_{i} p_{i} T_{i} (\mathbf{a}=\mathbf{a}_i \mid X).
\end{equation}
The above formally reproduces the essential definitions of PDBs as presented in \cite{Sa2019}. 
Here, we remark that in (\ref{Probqueryanswer}) the $T_{i}\left(\mathbf{a}=\mathbf{a}_i\mid X\right)$ are defined tablewise, i.e., for a particular set $X$. 
Moreover, tuple distributions $T_i$ are not necessarily the result of restricting a global deterministic assignment as happens for $D_{i}$ in (\ref{NCbehavior}) where indeed exist a global deterministic joint distribution $J$.

Now, the main question arises: \textit{Could behaviors generated by PDBs show contextuality?}, i.e., that no single table PDB with $X=V$ could reproduce all $P\left(\mathbf{a}\mid X\right)$?
We argue in the affirmative, and we illustrate an exemplary situation with a PDB generating the 3-cycle behavior presented in Section \ref{sec:unboundedness} (see Figure \ref{fig:PDBsContex}).

Note, we could also consider cases in which a violation of condition (\ref{consistencypairwise}) takes place; however, in such a situation, contextuality would be more trivial since the direct probability assignment of tables would disagree on common attributes, forcing a necessary split into a subset of PDBs. 
Such a situation would be easy to detect by applying analogous queries as in testing pairwise agreement of tables in relational databases, which can be answered polynomially in the number of tables \cite{Beeri1983, Abramsky2013}. 
For this reason, we assume that any mismatch is corrected, and we focus on the non-trivial case of a PDB equivalent to a collection $P\left(\mathbf{a} \mid X\right)$ which for all X satisfies the consistency condition (\ref{consistencypairwise}). 
The above PDB we denote as \textit{pairwise consistent}. 

\begin{figure*}[t]
    \includegraphics[width=\textwidth]{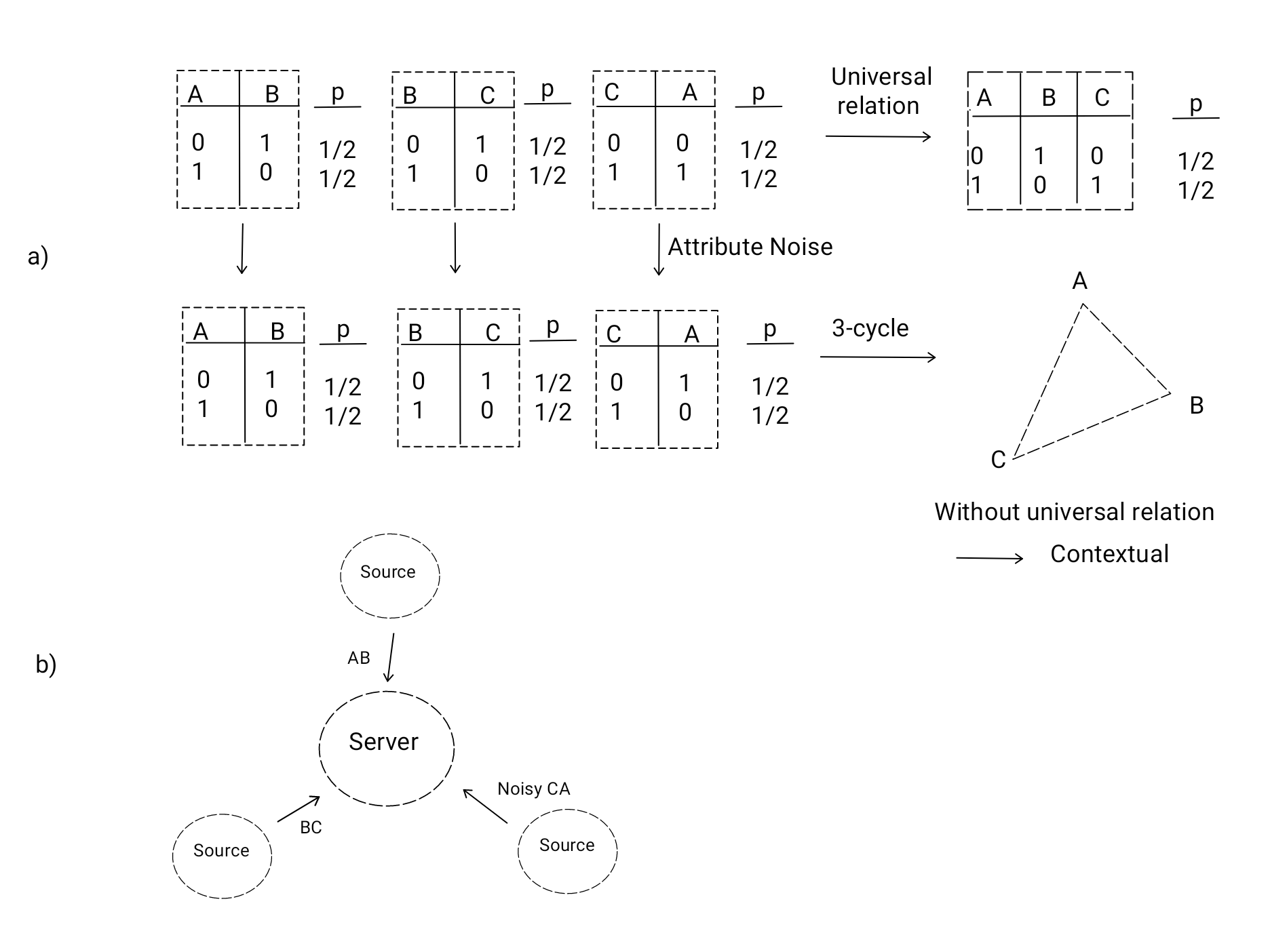}
    \caption{
        \label{fig:PDBsContex}
        (a) An example of how simple attribute noise could modify tuples in a source database to obtain a PDB that a single table PDB can not reproduce. 
        (b) A scheme in which a server makes up a noisy PDB by collecting data from sources at different network nodes. 
        In the previous typical situation, attribute noises due to transmission could lead to  PDBs equivalent to contextual behaviors like those studied in quantum contextuality. 
    }
\end{figure*}

A fundamental problem in managing PDBs is determining the source database when affected by noise, which can be classified into label noises or attribute noises \cite{ZhuW2004}. 
The former are modifications in the attributes designated to the tables, while the latter are modifications of the values associated with the attributes (tuple corruption), which are more common and harmful \cite{ZhuW2004, Nett2010}. 

It is standard to assume that a clean PDB contains tables that could be merged into a single table (non-contextual), of which the others turn out to be reductions (marginals) to a PDB for a subset of attributes. 
For example, in Figure \ref{fig:PDBsContex}a) we see a PDB with a database composed of three tables, each with attributes $X_{1}=\left\{ A B\right\}$, $X_{2}=\left\{ B C\right\} $, and $X_{3}=\left\{ C A\right\}$ and for each table we have the PDB description given by 
\begin{equation}
    \left\{ \left\{ T_{1}\left(\mathbf{a}=01\mid X_{j}\right),\frac{1}{2}\right\} ,\left\{ T_{2}\left(\mathbf{a}=10\mid X_{j}\right),\frac{1}{2}\right\} \right\}_{j=1}^{2},
\end{equation}
and 
\begin{equation}
    \left\{ \left\{ T_{1}\left(\mathbf{a}=00\mid X_{3}\right),\frac{1}{2}\right\} ,\left\{ T_{2}\left(\mathbf{a}=11\mid X_{3}\right),\frac{1}{2}\right\} \right\},
\end{equation}
such a PDB could be merged to a single table with $X=V=\left\{ ABC\right\}$ with a PDB: 
\begin{equation}
    \begin{split}
        \left\{ \left\{ T_{1}\left(\mathbf{a}=010\!\mid\!\left\{ ABC\right\} \right),\frac{1}{2}\right\},\right.\\
        \left. \left\{ T_{2}\left(\mathbf{a}=101\!\mid\!\left\{ ABC\right\} \right),\frac{1}{2}\right\} \right\}
    \end{split}
    \label{unirel}
\end{equation}
denoted \textit{universal relation} in database nomenclature \cite{page1990, Beeri1983}.
However, the situation can be drastically different in the case of unclean PDBs, which are the result of attribute noise affecting their separate tables. 
For example, consider a database distributed over a network and a server collecting the information for each table $X_j$ from a different source node Figure \ref{fig:PDBsContex}b). 
If, in such a case, a systematic noise affects the values of the tuples, as shown in Figure \ref{fig:PDBsContex}a), the PDB generated on the server would be 
\begin{equation}
    \left\{ \left\{ T_{1}\left(\mathbf{a}=01\mid X_{j}\right),\frac{1}{2}\right\} ,\left\{ T_{2}\left(\mathbf{a}=10\mid X_{j}\right),\frac{1}{2}\right\} \right\}_{j=1}^{3}
\end{equation}
which is equivalent to the 3-cycle contextual behavior.

To date, we are not aware of any study on the contextuality of PDBs in databases, which we believe is due to methodological contingencies: a) most noise studies are about label noise \cite{ZhuW2004, SAEZ2022}, while, as in our example we expect systematic attribute noises to generate the contextual behaviors b). 
Moreover, it is usual to assume that the attributes are weakly correlated \cite{Nett2010}, and the repair methods rule out strongly correlated attributes a priori \cite{kotu2018}, which hides any potential contextuality in the studied databases since non-trivial contextual behaviors happen only for strongly correlated attributes.

If realistic PDBs exhibit contextuality, the methods and results of our work would directly contribute to the management of unclean PDBs.
For example, the memory size of the automaton (\ref{automaton}) could be used as a benchmark to study storage overhead due to attribute noise, applying the rank of contextuality as a quantifier. 
Furthermore, minimal sets of states of the automaton's memory show potential repairs of the database, which we could classify according to prior knowledge of the noise. 
In fact, the universal relation (\ref{unirel}) is equivalent to the non-contextual behaviour (\ref{repair1}) of our example from Section \ref{sec:unboundedness}.
When managing PDBs in situations like in Figure \ref{fig:PDBsContex}b), choosing from sets of states like (\ref{repair1}) or (\ref{repair2}) will depend on the reliability we assigned to sources or transmission channels. 

Moreover, the construction of contextual behaviors such as those presented in equation (\ref{Contextforest}) allows us to understand the possible results of data corruption and their generating operations. 
Indeed, in the language of quantum resources, the previous noises would be instances of channels that generate contextuality. 
In this sense, the automaton realization of such behaviors and channels would allow noise modeling in databases, essential for controlling and cleaning PDBs.

The simplicity and naturalness with which a class of noise common in databases could generate a contextual behavior in PDBs is an indication that such a phenomenology should affect data mining in large-scale knowledge databases. 
Our conviction is that testing such a phenomenon in actual databases is a relevant challenge due to its potential consequences for the field of data management. 
However, the experimental demonstration of the previously described phenomenology goes beyond the scope of this work, and we consider it a goal for future research in collaboration with data scientists.

\subsection{Randomness Amplification}

The aim of the randomness amplification is to take outputs from some source of weak randomness and process it to obtain a more random sequence.
There are various models of a weak source of randomness, but we will focus here on the most famous one proposed by Santhat and Vazirani \cite{SV}.
The source is called $\varepsilon$-SV source.
We say that the source that produces sequence of binary outputs $S_1, S_2, \ldots$ is $\varepsilon$-SV source if
\begin{equation}
    \frac{1}{2} - \varepsilon \leq P(S_n \mid S_{n-1}, S_{n-2}, \ldots, S_1, e) \leq \frac{1}{2} + \varepsilon
\end{equation}
where $e$ represents an arbitrary random variable prior to $S_1$, 
The $\varepsilon \in [0, 1/2]$ is a parameter that describes how random the source is.
If $\varepsilon = 0$, then we obtain a truly random sequence.
On the other hand, if $\varepsilon = 1/2$, the source could even be deterministic.
We then can formally state that the goal of randomness amplification is to take some $\varepsilon$-SV source and process its output in such a way that we obtain a sequence that fulfills the definition of $\varepsilon'$-SV source for some $\varepsilon' < \varepsilon$.

Santha and Vazirani \cite{SV} proved that the randomness amplification of a single $\varepsilon$-SV source using a classical extractor is impossible.
However, in the seminal paper, Colbeck and Renner \cite{Colbeck2012} showed that the task can be achieved using quantum devices.
Their result opened a new field in quantum information theory and led to various new randomness amplification protocols.

One of the most important questions was: Is device-independent quantum randomness amplification possible when the weak source is correlated with the quantum device?
Wojewódka et al.\ \cite{Wojewodka2017} investigated this problem and proved that although some correlations are allowed (the one that fulfills the so-called SV-box condition), the amplification is impossible for arbitrary correlations. 
That impossibility comes from the fact that the adversary can, for each input, use one of the deterministic behaviors that are appropriate for that input and indistinguishable from the honest quantum implementation.  


It is an identical mechanism as the simulating contextual behavior by some number of non-contextual ones described in Section \ref{sec:introduction}.
Therefore, we can use the rank of contextuality to measure the required resources for such a form of attack. 
By resources, we mean here the memory and the amount of correlation. 
Furthermore, the rank of contextuality could help determine the number of resources needed for other attacks (such as side-channel attacks) on different device-independent protocols.

\section{Concluding remarks}
\label{sec:discussion}

We have introduced a measure of contextuality, the (log-)rank of contextuality, and confirm its novelty by proving that it differs from other known measures. 
Subsequently, we show that the RC measure is monotonic under a broad class of operations. 
Moreover, we provide a concrete example of (possibly supraquantum) contextual behaviors whose rank of contextuality is any natural number. 
As potential applications, we demonstrate that our measure naturally captures the cost of an attack by an adversary who correlates a source of weakly-private randomness with an amplifying device. 
Additionally, we considered possible contributions to the field of database management. 
In particular, we proposed a new avenue connecting databases and contextuality by showing that contextual models can be valuable semantics for probabilistic databases.
    
Specifically, we present a simple example where a probabilistic database presents contextuality generated by attribute noise on strongly correlated attributes. 
Our analysis suggests that models such as the automaton (\ref{automaton}) would allow us to simulate the generation of attribute noise and propose repair schemes beyond the usual assumption of weakly correlated attributes \cite{Nett2010, kotu2018}. 
In addition, such a model would quantify the storage overhead due to attribute noise by computing the rank of contextuality.
     
To confirm the presence of contextuality in actual probabilistic databases, one could test the correlations between attributes and determine if they violate non-contextual inequalities \cite{Cabello2008, Klyachko2008, Cabello2013}. 
This possibility suggests that quantum contextuality methods would detect particular data corruptions analogously to how non-locality techniques serve to detect causal connections in Bayesian networks \cite{Nery2018, Miklin2021}. 
The preceding suggests research in databases that we intend to develop in future works.
    
Finally, it would be interesting to construct a quantum behavior having the rank of contextuality different from the power of 2 (the power of 2 is achieved by a tensor product of the Peres-Mermin behaviors). 
It also seems attractive to investigate if the RC measure is more directly related to the memory of contextuality, going beyond the difference in the Peres-Mermin behavior.


\begin{acknowledgments}
    KH acknowledges National Science Centre, Poland, grant Sonata Bis 5, 2015/18/E/ST2/00327 and National Science Centre, Poland, grant OPUS 9, 2015/17/B/ST2/01945.
	MS acknowledges National Science Centre, Poland, grant SHENG 1, 2018/30/Q/ST2/00625.
	RS acknowledges financial support by the Foundation for Polish Science through TEAM-NET project (contract no.\ POIR.04.04.00-00-17C1/18-00).
	We acknowledge partial support by the Foundation for Polish Science (IRAP project, ICTQT, contract no.\ MAB/2018/5, co-financed by EU within Smart Growth Operational Programme). The 'International Centre for Theory of Quantum Technologies' project (contract no.\ MAB/2018/5) is carried out within the International Research Agendas Programme of the Foundation for Polish Science co-financed by the European Union from the funds of the Smart Growth Operational Programme, axis IV: Increasing the research potential (Measure 4.3). 
	
	%
		
	
	

\end{acknowledgments}

	
\bibliography{bibliography}

	


\end{document}